\documentclass[notitlepage,aip]{revtex4-1}
\usepackage{graphicx}
\usepackage{multirow}
\usepackage{physics}
\usepackage{xcolor}
\usepackage{amssymb}
\usepackage{tikz}
\usetikzlibrary{external}
\usepackage{amsthm}
\usepackage{amsmath}
\usepackage{hyperref}
\usepackage{lineno}

\newtheorem{prop}{Proposition}

\begin{document}

%
%
\title{Interactive Proofs of Quantumness via Mid-Circuit Measurements}

\author{Daiwei Zhu$^{1,2,9,12\dagger\ast}$, Gregory D. Kahanamoku-Meyer$^{3,4\dagger}$, Laura Lewis$^{5,6}$, Crystal Noel$^{1,7,8}$, Or Katz$^{7,8}$, Bahaa Harraz$^{1}$, Qingfeng Wang$^{1,2,11}$, Andrew Risinger$^{1,2}$, Lei Feng$^{1,2}$, Debopriyo Biswas$^{1,2}$, Laird Egan$^{1,2}$, Alexandru Gheorghiu$^{5,10}$, Yunseong Nam$^{9}$, Thomas Vidick$^{5}$, Umesh Vazirani$^{3,4}$, Norman Y.~Yao$^{3,4}$, Marko Cetina$^{1,7}$, Christopher Monroe$^{1,2,7,8,9}$ \\
\vspace{0.2in}
\normalsize{$^{1}$Joint Quantum Institute, Departments of Physics and Electrical and Computer Engineering, University of Maryland, College Park, MD 20742, USA}\\
\normalsize{$^{2}$Joint Center for Quantum Information and Computer Science, NIST/University of Maryland, College Park, MD 20742, USA}\\
\normalsize{$^{3}$Department of Physics, University of California, Berkeley, CA 94720, USA}\\
\normalsize{$^{4}$Materials Sciences Division, Lawrence Berkeley National Laboratory, Berkeley, CA 94720, USA}\\
\normalsize{$^{5}$Department of Computing and Mathematical Sciences, California Institute of Technology, CA 91125, USA}\\
\normalsize{$^{6}$Division of Physics, Mathematics, and Astronomy, California Institute of Technology, CA 91125-0001, USA}\\
\normalsize{$^{7}$Duke Quantum Center and Department of Physics, Duke University, Durham, NC 27708, USA}\\
\normalsize{$^{8}$Department of Electrical and Computer Engineering, Duke University, Durham, NC 27708, USA}\\
\normalsize{$^{9}$IonQ, Inc., College Park, MD  20740, USA}\\
\normalsize{$^{10}$Institute for Theoretical Studies, ETH Z{\"u}rich, CH 8001, Switzerland}\\
\normalsize{$^{11}$Chemical Physics Program and Institute for Physical Science and Technology, University of Maryland, College
Park, MD 20742, USA}
\normalsize{$^{12}$Departments of Electrical and Computer Engineering, University of Maryland, College Park, MD 20742, USA}\\
\normalsize{$^\ast$To whom correspondence should be addressed. E-mail:  daiwei@terpmail.umd.edu.}\\
\normalsize{$^\dagger$These authors contributed equally.}}

\date{\today}

\maketitle 

\textbf{Interaction is a powerful resource for quantum computation. It can be utilized in applications ranging from the verification of quantum algorithms all the way to verifying quantum mechanics itself. Here, we present the first implementation of interactive protocols for proofs of quantumness --- which when suitably scaled promise the efficient verification of quantum computational advantage. The key feature that distinguishes these protocols from existing demonstrations of quantum advantage is that the classical verification is efficient and scales polynomially rather than exponentially with the number of qubits. The experimental implementation of such interactive protocols requires the ability to independently measure subsets of qubits in the middle of a quantum circuit and to continue coherent evolution afterwards. This is achieved by spatially isolating target qubits via shuttling, and opens the door to a range of quantum interactive protocols as well as new information processing architectures.  }

To date, the field of experimental quantum computation has largely operated in a non-interactive paradigm, where classical data is extracted from the computation only at the very last step.
While this has led to many exciting advances, it has also become clear that 
in practice, interactivity---made possible by mid-circuit measurements performed on the quantum device---will be crucial to the operation of useful quantum computers.
For example, within quantum error correction, mid-circuit measurements are used to project the state onto a single error syndrome, which can then be corrected~\cite{nielsen_quantum_2010, ryananderson2021realization}.
Certain quantum machine learning algorithms also leverage mid-circuit measurements to introduce essential non-linearities~\cite{cong_quantum_2019}.
Recent work has shown that interaction can do much more: it has emerged as an indispensable tool for  verifying the behavior of untrusted quantum devices~\cite{mahadev2018classical,brakerski2018cryptographic,alex2019computationallysecure}, and even for testing the fundamentals of quantum mechanics itself~\cite{aharonov_interactive_2010}.


Ultimately, any quantum computer is an interaction between a classical machine and a much more powerful quantum system with its exponentially large Hilbert space.
Intuitively this asymmetry poses an insurmountable barrier for the classical machine's ability to certify the behavior of an untrusted quantum device. 
This challenge shadows one pursued in the field of classical computing, which asks whether a skeptical, computationally-bounded ``verifier,'' who is not powerful enough to validate a given statement on their own, can be convinced of its veracity by a more powerful but untrusted ``prover.''
Several decades ago, the field of classical complexity theory began to pursue this idea through a novel tool called an \emph{interactive proof}.
In these protocols, the verifier's goal is to accept only valid statements, regardless of whether the prover behaves honestly or attempts to cheat.
The crown jewel of computational complexity theory is a set of results showing that in certain scenarios multiple rounds of interaction allow the verifier to detect cheating by even \emph{arbitrarily computationally powerful} provers~\cite{goldwasser_knowledge_1989, lund_algebraic_1990, shamir_ip_1992}.
The key idea is that interaction can force the prover to \emph{commit} to some data early in the protocol, upon which the verifier follows up with queries that can only be answered consistently if the prover is being truthful.
In exciting recent developments, success has been found in the application of this idea to quantum computing: interactive proofs have been shown to allow the verification of a number of practical quantum tasks, including random number generation,~\cite{brakerski2018cryptographic} remote quantum state preparation,~\cite{alex2019computationallysecure} and  the delegation of  computations to an untrusted quantum server.\cite{mahadev2018classical}
Connecting to seminal recent sampling experiments~\cite{}, perhaps
the most direct application of an interactive protocol is for a ``proof of quantumness''---the classically-verifiable demonstration of non-classical behavior from a single quantum
device.




In practice, the experimental implementation of interactivity is extremely challenging.
It requires the ability to independently measure subsets of qubits in the middle of a quantum circuit and to continue coherent evolution afterwards. 
Unfortunately, the measurement of a target qubit typically disturbs neighboring qubits, degrading the quality of computations following the mid-circuit measurement. 
%
One solution, which finds commonality among multiple quantum platforms is to spatially isolate target qubits via shuttling~\cite{hensinger2021quantum,bluvstein2022quantum,pino2021demonstration}. 
While daunting from the perspective of quantum control, experimental progress toward coherent qubit shuttling opens the door not only to interactivity but also to distinct information processing architectures~\cite{kielpinski2002architecture}.
\begin{figure}
\centering
 \includegraphics[width=0.6\textwidth]{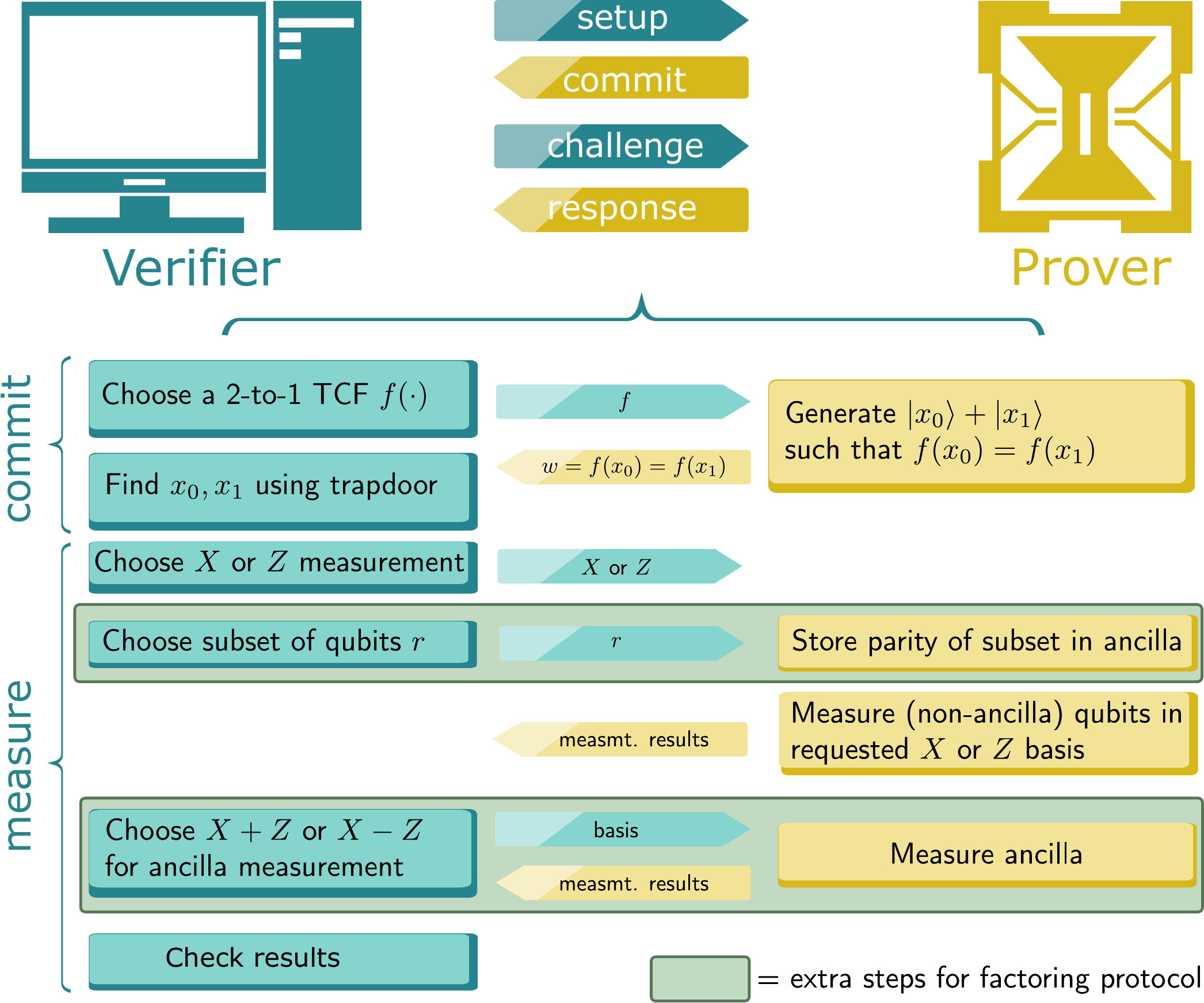}
 \caption{Schematic of an interactive quantum verification protocol. The verifier's goal is to test the ``quantumness" of the prover through an exchange of classical information.
 The protocol begins with the verifier sending the prover an instance of a trapdoor claw-free function. 
 By applying this function to a superposition of all possible inputs and projectively measuring the result, the prover commits to a particular quantum state $\ket{x_0} + \ket{x_1}$.
Subsequent challenges issued by the verifier specify how to measure this state and enable the efficient validation the prover's commitment.
 The LWE-based protocol requires two rounds of interaction, while the factoring-based protocol requires an additional round (green box).}
 \label{fig:interaction}
\end{figure}

In this work, we implement two complementary interactive proofs of quantumness, shown
 in the schematic of Fig.~\ref{fig:interaction}, on an ion trap quantum computer with up to 11 qubits and 145 gates.
The interactions between verifier and prover are enabled by the experimental realization of mid-circuit measurements on a portion of the qubits (Fig.~\ref{fig:shuttling})~\cite{ryananderson2021realization,pino2021demonstration,wan2019quantum}.
The first protocol involves two rounds of interaction and is based upon the learning with errors (LWE) problem \cite{regev2009lattices, regev2010learning}.
The LWE construction is unique because it exhibits a property known as the ``adaptive hardcore bit'' \cite{brakerski2018cryptographic}, which enables a particularly simple measurement scheme.
%
The second protocol circumvents the need for this special property and thus applies to a more general class of cryptographic functions; here we use a function from the Rabin cryptosystem \cite{rabin_digitalized_1979,goldwasser1988digital}. 
By using an additional round of interaction, the cryptographic information is condensed onto the state of a single qubit.
This makes it possible to implement an interactive proof of quantumness whose hardness is equivalent to that of factoring, but whose associated circuits can exhibit an asymptotic scaling much simpler than Shor's algorithm, with the number of required gates that is almost linear in the problem size. 

\section{Trapdoor Claw-free functions}

Both interactive protocols (Fig.~\ref{fig:interaction}) rely upon a cryptographic primitive called a \emph{trapdoor claw-free function} (TCF)~\cite{GoldwasserMR84}---a 2-to-1 function $f$ for which it is cryptographically hard to find two inputs mapping to the same output (i.e.~a ``claw'').
The function also has a ``trapdoor,'' a secret key with which it is easy to compute the pre-images $x_0$ and $x_1$ from any image $w = f(x_0) = f(x_1)$.
The key intuition behind the protocols is the following: Despite the claw-free property, a quantum computer can efficiently generate a \emph{superposition} of two pre-images that form a claw; this is most simply realized by evaluating $f$ on a superposition of the entire domain, and then collapsing to a single image, $w$, via measurement.
In this way, a quantum prover can generate the state $\ket{\psi} = (\ket{x_0}+\ket{x_1})\ket{w}$, where $w$ is the measurement result.
Note that the value of $w$ changes each time the protocol is executed,  making it impossible to learn any additional information by repeating the protocol.
After receiving $w$ from the prover, the verifier can use the trapdoor to compute $x_0$ and $x_1$, thus giving the verifier full knowledge of the prover's quantum state.
The verifier then asks the prover to measure $\ket{\psi}$.
In particular, they  request either a standard basis measurement (yielding $x_0$ or $x_1$ in full), or a measurement that \emph{interferes} the states $\ket{x_0}$ and $\ket{x_1}$.
In both cases, the verifier checks the measurement result on a per-shot basis.

\emph{The learning with errors problem}---It is believed to be classically intractable to recover an input vector from the result of certain noisy matrix-vector multiplications---this constitutes the LWE problem~\cite{regev2009lattices, regev2010learning}.  
In particular, a secret vector, $s \in \{0,1\}^n$, can be encoded into an output vector, $y = As + e$, where  $A \in \mathbb{Z}^{m \times n}_q$ is a matrix and $e$ is an error  vector corresponding to the noise.
Using the LWE problem, a TCF can be constructed as $f(b, x) = \lfloor Ax + b \cdot y \rceil$, where $b$ is a single bit that controls whether $y$ gets added to $Ax$ and $\lfloor \cdot \rceil$ denotes a rounding operation~\cite{banerjee2012pseudorandom, alwen2013learning} (see Supplementary Information for additional details). Here, $s$ and $e$ play the role of the trapdoor, and a claw corresponds to colliding inputs $\{(0,x_0),(1,x_1)\}$ with $f(0,x_0) = f(1,x_1)$ and $x_0 = x_1 +s$.
By  implementing the protocol described above and illustrated in Figure~\ref{fig:interaction}, the prover is able to generate the state $\ket{\psi} = (\ket{0,x_0}+\ket{1,x_1})\ket{w}$.
For the aforementioned ``interference'' measurement, the prover simply measures each qubit of the superposition in the $X$ basis.
Crucially, the result of this measurement is cryptographically protected by the adaptive hardcore bit property \cite{brakerski2018cryptographic}.

\emph{Rabin's function}---The function, $f(x) = x^2 \bmod N$, with $N$ being the product of two primes, was originally introduced in the context of digital signatures~\cite{rabin_digitalized_1979, goldwasser1988digital}. %
This function has the property that finding two colliding pre-images in the range $[0,N/2]$ is as hard as factoring $N$.
Moreover, the prime decomposition $N=pq$ can serve as a trapdoor, enabling one to invert the function for any output.
Thus, $f(x)$ is a trapdoor claw-free function.
However, $f(x)$ does not have the adaptive hardcore bit property, making the simple $X$-basis ``interference'' measurement (described in the LWE context above)  not provably secure. 
To get around this, we perform the ``interference'' measurement differently.
First, the verifier chooses a random subset of the qubits of the superposition, and the prover stores the parity of that subset on an ancilla.
Then, the prover measures everything except the ancilla in the $X$ basis; the polarization of this remaining ancilla qubit is  cryptographically protected, in the same way that the state of one half of a Bell pair is protected (i.e.~by ``no communication'') when the other half is measured.  
Following this intuition, the verifier requests a measurement of the ancilla qubit in the $Z+X$ or $Z-X$ basis, effectively completing the Bell test~\cite{bell_einstein_1964, clauser1969proposed}; the verifier accepts if the prover returns the more likely measurement outcome.
Crucially, the  dependence of the measurement result on the claw renders it  infeasible to guess classically  \cite{kahanamoku2021classically}.

\begin{figure}
\centering
 \includegraphics[width=0.8\textwidth]{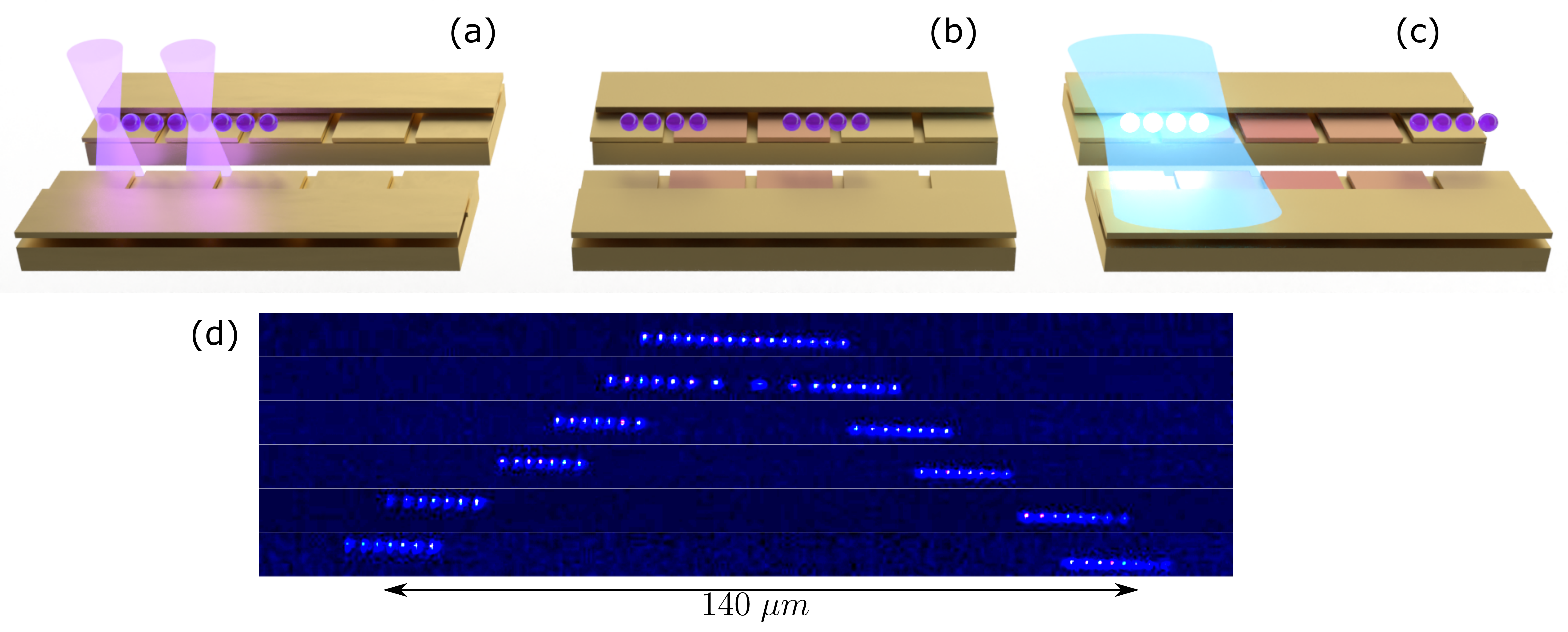}
 \caption{(a-c) Schematic illustration of our mid-circuit measurement protocol. (a) To start, the ions are closely-spaced in a 1D chain above a surface  trap. 
 Coherent gates are implemented via a combination of individual addressing beams (purple) and global beams (not shown). 
 (b) By tuning the electrodes of the surface trap, one can adjust the potential to deterministically split the ion chain.
 Depending on the protocol, we split the chain into either two or three individual segments. 
 We optimize the rate of shuttling to minimize the perturbation of the motional state. 
 (c) Once the segments are sufficiently far away from one another, it is possible to measure (blue beam) an individual segment without disturbing the coherence of the remaining ions. 
 After the measurement, the shuttling is reversed and the ion chain is recombined.
 (d) Fluorescence image of an example shuttling protocol for a chain of $N=15$ ions. At the start, the average spacing between ions is $\sim4 \mu$m. At the end of the splitting procedure, the distance between the two segments is $\sim550\mu$m um. Shown is the splitting up to a distance of $\sim140\mu$m until the two sub-chains reach the edge of the detection beam.
 }\label{fig:shuttling}
\end{figure}

\begin{figure}[t!]
\centering
 \includegraphics[width=\textwidth]{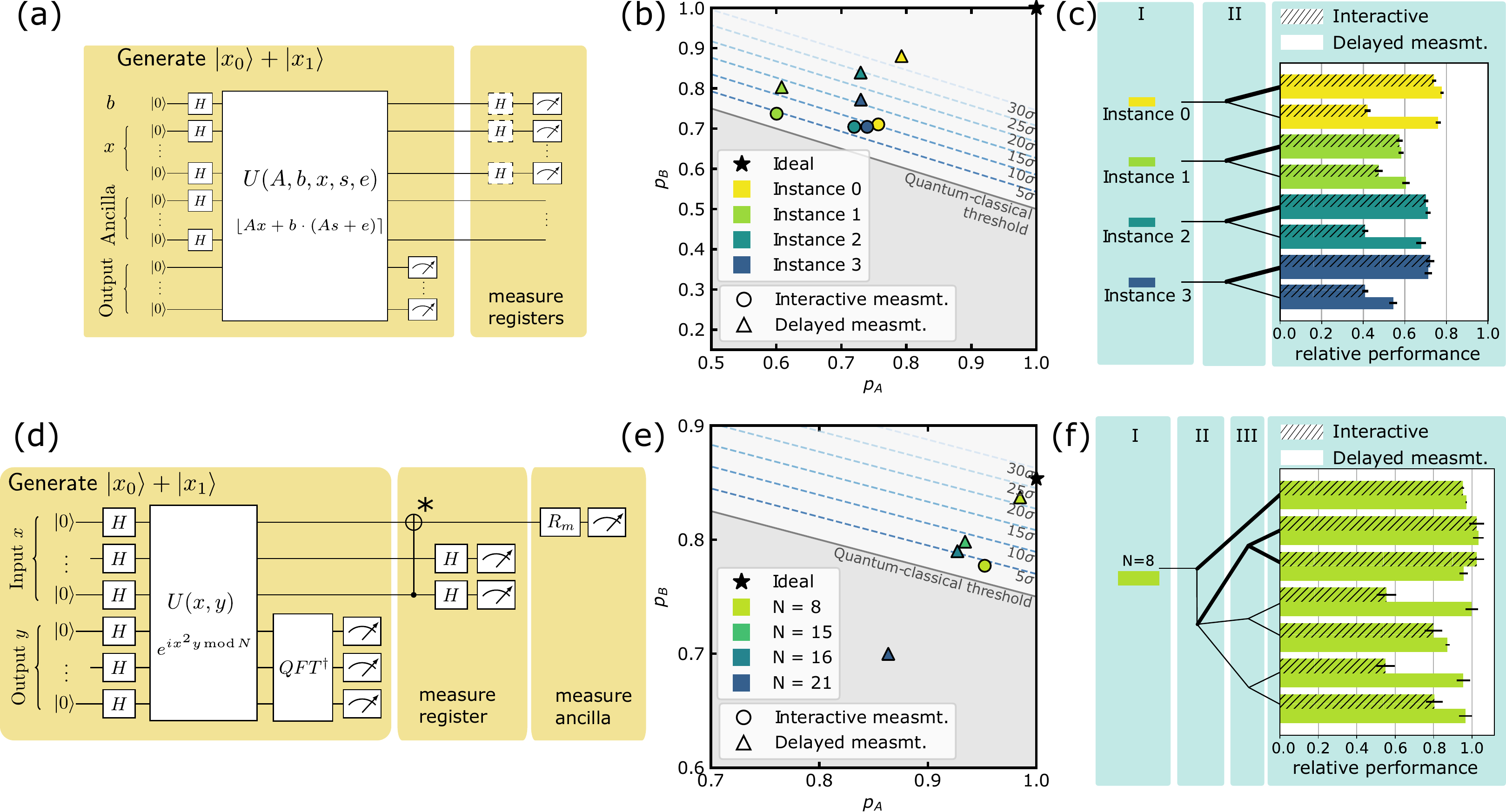}
 \caption{(a),(d) depict the circuit diagrams for the LWE- and factoring-based protocols, respectively.
 Details about the implementation of $U(A,b,x,y)$ and $U(x,y)$ are provided in the supplementary information.
 In (d), the CNOT gate marked with an asterisk stands for the operations needed to store the parity of selected qubits into the ancilla. 
 To reduce the impact of shuttling-induced gate fidelity degradation, we compute the parity for all of the verifier's possible selections and then choose the relevant one once the prover receives the challenge.
 (b),(e): Experimentally measured probabilities of passing the standard-basis ($p_A$) and interference measurement ($p_B$) challenges for the LWE- and factoring-based protocols. These probabilities are compared against the asymptotic classical limits ($p_A+2p_B \leq 2$ for LWE, derived in the supplementary information section~\ref{suppinfo:threshold}, and $p_A+4p_B \leq 4$ for factoring~\cite{kahanamoku2021classically}).
 Results for both interactive and delayed-measurement version of the protocols are presented. 
 Numerical values of $p_A$ and $p_B$ for each experiment, and the corresponding values of statistical significance, are provided in the supplemental information.
 (c),(f): The relative performance, $R$, of the experiments for all possible branches. Certain branches (thick lines) are robust to phase errors and exhibit similar performance for both interactive and delayed-measurement protocols. 
 }\label{fig:all_combined}
\end{figure}

\section{Implementing an interactive proof of quantumness}

In order to implement an interactive proof of quantumness, we design quantum circuits for both the LWE- and factoring-based protocols. %
 The high-level circuit diagrams are shown in Figs.~\ref{fig:all_combined}(a,b). 
 In both cases, the circuits are composed of several sections. First, the prover creates a uniform superposition $\ket{\psi} = \sum_{x=0}^{2^n-1} \ket{x}$ via  Hadamard gates, where $n$ is the number of input qubits.
 Then, they compute the TCF on an output register using this superposition as input [Fig.~\ref{fig:all_combined}(a,d)], thereby generating the state $\ket{\psi} = \sum_x \ket{x}  \ket{f(x)}$.
 Next, the prover performs a mid-circuit measurement  on the output register, collapsing the state to $\ket{\psi} =  \left (\ket{x_0} + \ket{x_1} \right )  \ket{w}$.
 Finally, based on the verifier's choice of measurement scheme (i.e.~standard vs.~interference), the prover must perform additional coherent gates and measurements (see Methods for a full description of the quantum circuits used). 

We implement both interactive protocols using an ion trap quantum computer, with a base chain length of 15 ions (Fig.~\ref{fig:shuttling}); for each $^{171}$Yb$^{+}$ ion, a qubit is encoded in a pair of hyperfine levels~\cite{monroe2021programmable}.
The quantum circuits are implemented via the consecutive application of native single and two-qubit gates using individual optical addressing [Fig.~\ref{fig:shuttling}(a)]~\cite{egan2020fault}. 
In order to realize rapid successive two-qubit interactions, we position the ions in a single, closely-spaced linear chain [Fig.~\ref{fig:shuttling}(d)]. 

This geometry makes it challenging to implement
mid-circuit measurements: light scattered from nearby ions during a  state-dependent fluorescence measurement can destroy the  state of the other ions. 
To overcome this issue, we vary the voltages on the trap electrodes to split and shuttle the ion chain, thereby spatially isolating the ions not being measured (Fig.~\ref{fig:shuttling}a-c).
Depending on the protocol, the ion chain is split into either two or three segments.
To measure the ions in a particular segment, we re-shape the electric potential to align the target segment with the detection system.
In addition, we calibrate and correct for spatial drifts of the optical beams, variations of stray fields, and unwanted phase accumulation during shuttling (see supplementary information sections~\ref{suppinfo:ionqc},~\ref{suppinfo:shuttling} for additional details). 

In this demonstration, the qubits play the role of the prover and the classical control system plays the role of the verifier.
This allows us to compile the decisions of the verifier into the classical controller prior to execution of the quantum circuit.

\section{Beating the classical threshold}
\label{sec:threshold}

Much like a Bell test, even a classical prover can pass the verifier's challenges with finite  probability.
Assuming that the trapdoor claw-free function is secure, this probability can be bounded by an asymptotic ``classical threshold''---which a quantum prover must exceed to demonstrate  advantage.
For both protocols, this threshold is best expressed in terms of the probabilities of passing the verifier's ``standard basis'' and ``interference'' checks, which we denote as $p_A$ and $p_B$, respectively (see supplementary material Section~\ref{suppinfo:verifier-checks} for the definition of the verifier's checks).
For the LWE-based protocol, the classical threshold is given by $p_A + 2p_B - 2 \leq \epsilon$ (derivation in Supplementary Information); for the factoring-based protocol, it is given by  $p_A + 4p_B - 4 \leq \epsilon$.~\cite{kahanamoku2021classically}
In both cases, $\epsilon$ is a function which goes to zero exponentially in the problem size.
An intuition for the difference between the thresholds is that the factoring-based protocol requires an additional round of interaction during the ``interference'' test.

As depicted in Figure~\ref{fig:all_combined}(b), we perform multiple instances of the LWE-based protocol for different matrices $A$ and noise vectors $e$.
For each of the verifier's possible choices, we repeat the experiment $\sim10^3$ times to collect statistics.  
This yields the experimental probabilities $p_A$ and $p_B$, allowing us to confirm that the quantum prover exceeds the asymptotic classical threshold in all cases.
The statistical significance by which the bound is exceeded (more than $6\sigma$ in all cases, see Table~\ref{tab:lwe_result} in the supplementary information) is shown in Figure 3(b). 
Figure~\ref{fig:all_combined}(e) depicts the analogous results for the factoring-based protocol, where the different instances correspond to different values of $N$.
For all but $N=21$, which requires the deepest circuit, the results exceed the asymptotic classical bound with more than $4\sigma$ statistical significance.
We utilize an error-mitigation strategy based on excluding iterations where $w$ is measured to be invalid, i.e.~not in the range of $f$ (see~\ref{suppinfo:postselection}); effectively, this implements a post-selection which suppresses  bit-flip errors~\cite{kahanamoku2021classically}. 

To further analyze the performance of each branch of the interactive protocol, corresponding to the verifier's choices [Figs.~\ref{fig:all_combined}(c,f)], we define the relative performance $R = (p_{\textrm{exp}}-p_{\textrm{guess}})/(p_{\textrm{ideal}}-p_{\textrm{guess}})$ for each branch, where $p_{\textrm{ideal}}$ is the probability that an error-free quantum prover would pass, $p_{\textrm{guess}}$ is the probability that a random guesser would pass, and $p_{\textrm{exp}}$ is the passing rate measured in the experiment. For a perfect quantum prover $R=1$, and for a device with zero fidelity, $R=0$.

For the LWE-based protocol, there are two rounds of interaction, corresponding to the two branches, I and II shown in Fig.~\ref{fig:all_combined}(c), while for the factoring-based protocol there are three rounds of interaction [Fig.~\ref{fig:all_combined}(f)].
By comparing the relative performance between the interactive and delayed-measurement versions of our experiment, we are able to probe a subtle feature of the protocols---namely, that certain branches  are robust to additional decoherence induced by the mid-circuit measurements. 
Microscopically, this robustness arises because these branches (thick lines, Figs.~\ref{fig:all_combined}c,f) do not depend on the phase coherence between  $\ket{x_0}$ and $\ket{x_1}$.
In particular, this is true for the standard-basis measurement branches in both protocols, and also for the branches of the factoring-based protocol where the ancilla is polarized in the $Z$ basis (see supplementary information section~\ref{suppinfo:factoringproto}).
Noting that mid-circuit measurements are expected to induce mainly phase errors, one would predict that those branches insensitive to phase errors should yield similar performance in both the    interactive and delayed-measurement cases. 
This is indeed borne out by the data. %

\emph{Discussion and Outlook}---The most direct application of our protocols is to validate an untrusted quantum device. 
Expanding the notion of an “untrusted prover,” one can use interactive protocols to
test quantum mechanics itself, demonstrating quantum computational advantage. 
The “statement” to be proven in this case is that quantum mechanics operates as  expected, even when scaled to system sizes that are impossible to classically simulate. 
As aforementioned, this idea is connected to  recent sampling experiments, which have demonstrated the system sizes and fidelities necessary to make classical simulation extremely hard or impossible~\cite{Arute2019, zhong_quantum_2020, wu2021strong, zhu2021quantum, aaronson_computational_2011,  lund2017quantum, harrow2017quantum, boixo2018characterizing, bouland_complexity_2019, aaronson2019classical}; however, these approaches do not operate as proofs of quantumness since there is no method to  efficiently  verify the output.
Furthermore, practical strategies for a classical impostor to replicate the sampling are still being explored.~\cite{huang_classical_2020, pan_simulating_2021, gray_hyper-optimized_2021, pan_solving_2021, yong_closing_2021, liu_redefining_2021, gao_limitations_2021}
To this end, as proofs of quantumness, interactive protocols satisfy three important criteria: 1) efficient (polynomial-time) classical verification, 2) classical hardness based on extremely well-studied cryptographic assumptions, and 3) the potential for realization on near-term quantum devices.~\cite{brakerski2018cryptographic, kahanamoku2021classically}

Our work represents the first experimental demonstration of an interactive protocol for a proof of quantumness and opens the door to a number of intriguing directions.
First, by scaling up our proposed circuits [Fig.~\ref{fig:all_combined}(a,b)], one should be able to perform a verifiable test of quantum advantage using $\sim 1600$ qubits (see supplementary information section~\ref{suppinfo:estimates}). 
At these scales, the challenge on near-term devices will almost certainly be the circuit depth; interestingly, recent advances suggest that our interactive protocols can be performed in constant  depth at the cost of a larger number of qubits~\cite{hirahara2021test, liu2021depth}.
Second, a clear next step is to apply the power of quantum interactive protocols to achieve more than just quantum advantage---for example, pursuing such tasks as certifiable random number generation, remote state preparation and the verification of arbitrary quantum computations~\cite{brakerski2018cryptographic, alex2019computationallysecure,  mahadev2018classical}.
In another path forward, one can imagine generalizing our experiment to include interactions with a remote verifier, for example over the internet, which could serve as a loophole-free way of remotely benchmarking quantum cloud services.
Finally, the advent of mid-circuit measurement capabilities in a number of platforms~\cite{pino2021demonstration,wan2019quantum,corcoles2021exploiting,rudinger2021characterizing}, enables the  exploration of new phenomena such as  entanglement phase transitions~\cite{skinner2019measurement,li2018quantum,noel2021observation} as well as the demonstration of coherent feedback protocols including quantum error correction~\cite{ryananderson2021realization}.  

\section{Acknowledgements}

The authors are grateful to Vivian Uhlir for the design of the verifier and prover figures.
This work is supported by the ARO through the IARPA LogiQ program,  the U.S. Department of Energy, Office of Science, National Quantum Information Science Research Centers, Quantum Systems Accelerator (QSA),  the AFOSR MURIs on Quantum Measurement/Verification and Quantum Interactive Protocols (FA9550-18-1-0161) and Dissipation Engineering in Open Quantum Systems, the NSF STAQ Program, the ARO MURI on Modular Quantum Circuits, the DoE ASCR Accelerated Research in Quantum Computing program (award No. DE-SC0020312), the AFOSR YIP award number FA9550-16-1-0495, the NSF QLCI program through grant number OMA-2016245, the IQIM, an NSF Physics Frontiers Center (NSF Grant PHY-1125565), the Gordon and Betty Moore Foundation (GBMF-12500028), the Dr. Max R{\"o}ssler, the Walter Haefner Foundation and the ETH Z{\"u}rich Foundation,  the NSF award DMR-1747426, a Vannever Bush Faculty Fellowship, the Office of Advanced Scientific Computing Research, under the Accelerated Research in Quantum Computing (ARQC) program,  the A. P. Sloan foundation and the David and Lucile Packard Foundation.
\textbf{Competing interests:} C.M. is Chief Scientist for IonQ, Inc. and has a personal financial interest in the company.

\section{Supplementary Materials}

\subsection{Result data} \label{suppinfo:result-data}

In Tables~\ref{tab:chsh_result} and \ref{tab:lwe_result} we present the numerical results for each configuration of the experiment, along with the number of samples obtained ($N_A$ and $N_B$), the measure of quantumness $q$, and the statistical significance of the result (see supp. info. Section~\ref{suppinfo:significance} for a description of how the significance is calculated).

We note that for the computational Bell test protocol, the sample size $N_B$ is less than the actual number of shots that passed postselection (ultimately leading to slightly less statistical significance than might otherwise be expected).
This is because the sample size varied for different values of the verifier's string $r$, yet we are interested in the passing rate $p_B$ averaged uniformly over all $r$ (not weighted by number of shots).
To account for this, we simply took the $r$-value with the fewest number of shots, and computed $N_B$ as if every $r$ value had had that sample size (even if some values of $r$ had more).

We also note that in some cases the statistical significance denoted here may be higher than that visually displayed in Figure~\ref{fig:all_combined} of the main text; this is because the contour lines in that figure correspond to the configuration with the smallest sample size.

\begin{table}[h!]
	\centering
	{\setlength{\tabcolsep}{0.5em}
    \begin{tabular}{c|c|c|c|c|c|c|c}
       N & Measurement scheme & $p_A$ & $p_B$ & $N_A$ & $N_B$ & Quantumness $q$ & Stat. significance \\
       \hline\hline
       8 & interactive & 0.952 & 0.777 & 4096 & 15267 & 0.061 & $4.3\sigma$ \\
       8 & delayed & 0.985 & 0.837 & 2736 & 17361 & 0.334 & $24.1\sigma$ \\
       15 & delayed & 0.934 & 0.798 & 2361 & 31353 & 0.127 & $10.0\sigma$ \\
       16 & delayed & 0.927 & 0.790 & 3874 & 53550 & 0.087 & $8.8\sigma$ \\
       21 & delayed & 0.864 & 0.700 & 2066 & 27944 & -0.338 & --- \\
    \end{tabular}}
	\caption{Results for various configurations of the computational Bell test protocol. For this protocol $q = p_A + 4p_B - 4$.} \label{tab:chsh_result}
\end{table}

\begin{table}[h!]
	\centering
	{\setlength{\tabcolsep}{0.5em}
    \begin{tabular}{c|c|c|c|c|c|c|c}
       Instance & Measurement scheme & $p_A$ & $p_B$ & $N_A$ & $N_B$ & Quantumness $q$ & Stat. significance \\
       \hline\hline
       0 & interactive & 0.757 & 0.710 & 8000 & 13381 & 0.178 & $18.6\sigma$ \\
       0 & delayed & 0.793 & 0.880 & 10000 & 9415 & 0.553 & $60.3\sigma$ \\
       1 & interactive & 0.601 & 0.737 & 8000 & 7622 & 0.075 & $6.2\sigma$ \\
       1 & delayed & 0.608 & 0.803 & 8000 & 7547 & 0.215 & $18.0\sigma$ \\
       2 & interactive & 0.720 & 0.704 & 14000 & 15310 & 0.129 & $15.0\sigma$ \\
       2 & delayed & 0.730 & 0.839 & 4000 & 3735 & 0.409 & $24.6\sigma$ \\
       3 & interactive & 0.740 & 0.704 & 8000 & 15189 & 0.148 & $16.2\sigma$ \\
       3 & delayed & 0.730 & 0.772 & 8000 & 7528 & 0.274 & $23.1\sigma$ \\
    \end{tabular}}
	\caption{Results for various configurations of the LWE-based protocol. For this protocol $q = p_A + 2p_B - 2$.} \label{tab:lwe_result}
\end{table}

\subsection{Trapped Ion Quantum Computer} \label{suppinfo:ionqc}
The trapped ion quantum computer used for this study was designed, built, and operated at the University of Maryland and is described elsewhere \cite{egan2020fault, cetina2020axial}. The system consists of a chain of fifteen single $^{171}$Yb$^+$ ions confined in a Paul trap and laser cooled close to their motional ground state. Each ion provides one physical qubit in the form of a pair of states in the hyperfine-split $^2S_{1/2}$ ground level with an energy difference of 12.642821 GHz, which is insensitive to magnetic fields to first order. The qubits are collectively initialized through optical pumping, and state readout is accomplished by state-dependent flourescence detection \cite{olmschenk07}. Qubit operations are realized via pairs of Raman beams, derived from a single 355-nm mode-locked laser \cite{debnath2016demonstration}. These optical controllers consist of an array of individual addressing beams and a counter-propagating global beam that illuminates the entire chain. Single qubit gates are realized by driving resonant Rabi rotations of defined phase, amplitude, and duration. Single-qubit rotations about the z-axis, are performed classically with negligible error. Two-qubit gates are achieved by illuminating two selected ions with beat-note frequencies near motional sidebands and creating an effective Ising spin-spin interaction via transient entanglement between the two ion qubits and all modes of motion \cite{molmer99,solano99,milburn00}. To ensure that the motion is disentangled from the qubit states at the end of the interaction, we used a pulse shaping scheme by modulating the amplitude of the global beam \cite{choi2014optimal}.

\subsection{Verifier's checks} \label{suppinfo:verifier-checks}

In this section we explicitly state the checks performed by the verifier to decide whether to accept or reject the prover's responses for each run of the protocol.
We emphasize that these checks are performed on a per-shot basis, and the empirical success rates $p_A$ and $p_B$ are define as the fraction of runs (after postselection, see below) for which the verifier accepted the prover's responses.

For both protocols, the ``A" or ``standard basis" branch check is simple.
The prover has already supplied the verifier with the image value $w$; for this test the prover is expected to measure a value $x$ such that $f(x) = w$.
Thus in this case the verifier simply evaluates $f(x)$ for the prover's supplied preimage $x$ and confirms that it is equal to $w$.

For the ``B" or ``interference" measurement, the measurement scheme and verification check is different for the two protocols.
For the LWE-based protocol, the interference measurement is an $X$-basis measurement of all of the qubits holding the preimage superposition $\ket{x_0} + \ket{x_1}$.
This measurement will return a bitstring $d$ of the same length as the number of qubits in that superposition, where for each qubit, the corresponding bit of $d$ is 0 if the measurement returned the $\ket{+}$ eigenstate and 1 if the measurement returned the $\ket{-}$ eigenstate.
The verifier has previously received the value $w$ from the prover and used the trapdoor to compute $x_0$ and $x_1$; the verifier accepts the string $d$ if it satisfies the equation
\begin{equation} \label{eq:lwe-interference-check}
d \cdot x_0 = d \cdot x_1
\end{equation}
where $(\cdot)$ denotes the binary inner product, i.e. $a\cdot b = \sum_i a_i b_i \mod 2$.
It can be shown that a perfect (noise-free) measurement of the superposition $\ket{x_0} + \ket{x_1}$ will yield a string $d$ satisfying Eq.~\ref{eq:lwe-interference-check} with probability 1.

The interference measurement for the computational Bell test involves a sequence of two measurements (in addition to the first measurement of the string $w$).
The first measurement yields a bitstring $d$ as above.
After performing that measurement, the prover holds the single-qubit state $(-1)^{d \cdot x_0} \ket{r \cdot x_0} + (-1)^{d \cdot x_0} \ket{r \cdot x_0}$, where $(\cdot)$ is the binary inner product as above and $r$ is a random bitstring supplied by the verifier. 
This state is one of $\{\ket{0}, \ket{1}, \ket{+}, \ket{-}\}$, and is fully known to the verifier after receiving $d$ (via use of the trapdoor to compute $x_0$ and $x_1$).
The second measurement is of this single qubit, in an intermediate basis $Z+X$ or $Z-X$ chosen by the verifier.
For any of the four possible states, one eigenstate of the measurement basis will be measured with probability $\cos^2 (\pi/8) \approx 85\%$ (with the other having probability $\sim 15\%$), just as in a Bell test.
The verifier accepts the measurement result if it corresponds to this more-likely result; an ideal (noise-free) prover will be accepted with probability $\sim 85\%$ (see Figure~\ref{fig:all_combined} of the main text). 

\subsection{Post-selection} \label{suppinfo:postselection}
Both the factoring-based and LWE-based protocols involve post-selection on the measurement results throughout the experiment.

\begin{table}[h!tbp]
    \centering
    \begin{tabular}{c|c|c}
       Instance & Delayed Measurement & Interactive Measurement\\
       \hline\hline
       0 & 3753/4000 & 13381/14000\\
       1 & 7547/8000 & 7622/8000 \\
       2 & 3735/4000 & 15310/16000\\
       3 & 7528/8000 & 15144/16000\\
    \end{tabular}
    \caption{This table displays the fractions of runs kept during post-selection for the LWE-based protocol, in the ``interference'' measurement branch. All runs are kept for the standard basis measurement.}
    \label{tab:postselection_LWE}
\end{table}

\begin{table}[h!tbp]
    \centering
    \begin{tabular}{c|c|c|c}
        N & Interactive & Branch & Runs kept/Total \\
        \hline\hline
        8 & Yes & A & 4096/9000 \\
        8 & Yes & B, r=01 & 5093/12000 \\
        8 & Yes & B, r=10 & 5089/12000 \\
        8 & Yes & B, r=11 & 5492/12000 \\
        \hline
        8 & No & A & 2736/6000 \\
        8 & No & B, r=01 & 5787/12000 \\
        8 & No & B, r=10 & 5818/12000 \\
        8 & No & B, r=11 & 5865/12000 \\
        \hline
        15 & No & A & 2361/6000 \\
        15 & No & B, r=001 & 4636/12000 \\
        15 & No & B, r=010 & 4532/12000 \\
        15 & No & B, r=011 & 4666/12000 \\
        15 & No & B, r=100 & 4496/12000 \\
        15 & No & B, r=101 & 4727/12000 \\
        15 & No & B, r=110 & 4479/12000 \\
        15 & No & B, r=111 & 4673/12000 \\
    \end{tabular}
    \hspace{0.5cm}
    \begin{tabular}{c|c|c|c}
        N & Interactive & Branch & Runs kept/Total \\
        \hline\hline
        16 & No & A & 3874/6000 \\
        16 & No & B, r=001 & 7842/12000 \\
        16 & No & B, r=010 & 7847/12000 \\
        16 & No & B, r=011 & 7732/12000 \\
        16 & No & B, r=100 & 7936/12000 \\
        16 & No & B, r=101 & 7870/12000 \\
        16 & No & B, r=110 & 7841/12000 \\
        16 & No & B, r=111 & 7650/12000 \\
        \hline
        21 & No & A & 2066/6000 \\
        21 & No & B, r=001 & 3992/12000 \\
        21 & No & B, r=010 & 4273/12000 \\
        21 & No & B, r=011 & 4137/12000 \\
        21 & No & B, r=100 & 4182/12000 \\
        21 & No & B, r=101 & 4193/12000 \\
        21 & No & B, r=110 & 4261/12000 \\
        21 & No & B, r=111 & 4221/12000 \\
    \end{tabular}
    \caption{Fraction of runs kept during postselection for each branch of the factoring-based protocol.}
    \label{tab:postselection_X2modN}
\end{table}

For the factoring-based protocol, this post-selection is performed on the measured value of the output register $w$. Due to quantum errors in the experiment, in practice it is possible to measure a value of $w$ that does not correspond to any preimages of the TCF---that is, there do not exist $x_0, x_1$ for which $f(x_0) = f(x_1) = w$, due to noise. Because such a result would not be possible without errors, measuring such a value indicates that a quantum error has occurred\cite{kahanamoku2021classically}. Thus, we perform post-selection by discarding all runs for which the measured value $w$ does not have two pre-images.

On the other hand, for the LWE protocol, we post-select in order to satisfy the conditions for the adaptive hardcore bit property to hold, as without this property, the protocol could be susceptible to attacks. In particular, the adaptive hardcore bit property requires that the result obtained from measuring the $x$ register using the ``interference'' measurement scheme be a nonzero bitstring \cite{brakerski2018cryptographic}. Hence, we simply post-select on this condition for the LWE case. Tables~\ref{tab:postselection_LWE} and \ref{tab:postselection_X2modN} explicitly show how many runs are kept using each post selection scheme.

We note that in both cases, post-selection does not affect the soundness of the protocols.
We only require that a non-negligible fraction of runs pass post-selection (to give good statistical significance for the results). 
This is indeed the case for our experiment, as can be seen in Tables~\ref{tab:postselection_LWE},~\ref{tab:postselection_X2modN}, as well as the statistical significance of the results in Tables~\ref{tab:chsh_result},~\ref{tab:lwe_result}.

\subsection{Shuttling and Mid-circuit measurements} \label{suppinfo:shuttling}

We control the position of the ions and run the split and shuttling sequences by changing the electrostatic trapping potential in a microfabricated chip trap \cite{osti_1237003} maintained at room-temperature. We generate 40 time-dependent signals using a multi channel DAC voltage source, which controls the voltages of the 38 inner electrodes at the center of the chip and the voltages of two additional outer electrodes. Owing to the strong radial confining potential used (with secular trapping frequencies near $3$ MHz), the central electrodes' potential effect predominantly the axial trapping potential, and in turn, generate movement predominantly along the linear trap axis. To maintain the ions at a constant height above the trap surface, we simulate the electric field based on the model in Ref.~\cite{osti_1237003}, and compensate for the average variation of its perpendicular component by controlling the voltages of the outer two electrodes.

In the first sequence, we split the 15 ion chain into two sub-chains of 7 and 8 ions, and shuttle the 8-ion group to $x=0.55$~mm away from the trap center at $x=0$. We then align the 7-ion chain with the individual-addressing Raman beams for the first mid-circuit measurement.
For the LWE-based protocol, we then reverse the shuttling process and re-merge the ions to a 15-ion chain, completing the circuit and performing a final measurement. For the factoring-based protocol, we shuttle instead the 8-ion sub-chain to the trap center and the 7-ion sub-chain to $x=-0.55$~mm. We then split this chain into 5- and 3-ion sub-chains, shuttle the latter to $x=0.55$ mm, and align the 5 ions at the center with the Raman beams for additional gates and a second mid-circuit measurement. Finally, we move away the measured ions and align the 3-ion group to the center of the trap to complete the protocol. Reversing of the sequence then prepares the ions in their initial state. For each protocol, all branches use the same shuttling sequences yet differ in the qubit assignment and the realized gates. The mid-circuit measurement duration was experimentally determined prior the experiment by maximizing the average fidelity of a Ramsey experiment using single-qubit gates, approximately optimizing for the trade-off between efficient detection of each sub-chain and stray light decoherence.

To enable efficient performance of the split and shuttling sequences we numerically simulate the electrostatic potential and the motional modes of the ions that are realized in the sequences. We minimize heating of the axial motion from low-frequency electric-field noise by ensuring that the calculated lowest axial frequency does not go below $>100$ KHz. We also minimize frequent ions loss due to collisions with background gas by maintaining a calculated trap depth of at least $20$~meV for each of the sub-chains throughout the shuttling sequences.
The simulations enable efficient alignment of the sub chains with the Raman beams, taking into account the variation of the potential induced by all electrodes.

We account and correct for various systematic effects and drifts which appear in the experiment. To eliminate the effect of systematic variation of the optical phases between the individual beams on the ions, we align each ion with the same individual beam throughout the protocol. Prior the experiment, we run several calibration protocols which estimate the the electrostatic potential at the center of the trap through a Taylor series representation up to a fourth order, estimating the dominant effect of stray electric-fields on the pre-calculated potential. We then cancel the effect of these fields using the central electrodes during the alignment and split sequences, as they are most sensitive to the exact shape of the actual electrostatic potential. Additionally, we routinely measure the common-mode drift of the individually addressing optical Raman beams along the linear axis of the trap and correct for them by automatic re-positioning of the ions through variation of the potential.

During shuttling, the ions traverse  an inhomogeneous magnetic field and consequently, each ion spin acquires a shuttling-induced phase $\phi_s^{(i)}$ which depends on its realized trajectory. We calibrate this  by performing a Ramsey sequence in which each qubit is put in a superposition of $(|0\rangle_i+|1\rangle_i)/\sqrt{2}$ before shuttling, and after the shuttling apply $R_x^{(i)}(\pi/2)R_z^{(i)}(\phi)$ gates where $\phi$ is scanned between $0$ to $2\pi$. Fitting the realized fringe for each ion enables estimation of the phases $\phi_s^{(i)}$, which are corrected in the protocols by application of the inverse operation $R_z^{(i)}(-\phi_s^{(i)})$ after shuttling.

\subsection{Circuit construction of the factoring-based protocol }
\label{suppinfo:factoringproto}

In this section, we describe the procedure for generating a superposition of the claw $\ket{x_0}+\ket{x_1}$ in the factoring-based protocol, as shown in Fig.~\ref{fig:all_combined}(a) of the main text. 

This is achieved by generating \begin{align}
    \sum_{0\leq x \leq N/2} \frac{1}{\sqrt{2^{N/2}}}\ket{x}\ket{f(x)=x^2\, mod \, N}.
\end{align}
and then measuring the $y=\ket{f(x)}$ register. We calculate $f(x)$ using a unitary $U(x,y)$ to encode the function into the phase of the $y$ register and applying a Inverse Quantum Fourier Transform (QFT$^\dagger$) to extract the result.

To start, we apply Hadamard gates to all qubits to prepare a uniform superposition of all the possible bit strings for the $x$- and $y$- registers:

\begin{align}
     \sum_{0\leq x \leq N/2,0\leq y \leq N} \alpha\ket{x}\ket{y},
\end{align}\label{eq:superposition}
where $\alpha$ is the normalization factor.

Next, we evolve the state with the unitary $U(x,y)=e^{2 \pi i \frac{x^2 y}{N}}$. Since the phase has period $2\pi$, the unitary is equivalent to $U(x,y)=e^{2 \pi i \frac{x^2 y \, mod\, N}{N}}$. We now show how to efficiently implement $U(x,y)=e^{2 \pi i \frac{x^2 y}{N}}$ on the ion trap quantum computer.

First, note the multiplication in the phase can be expressed as a sum of bit-wise multiplication
\begin{align}
U(x,y)=\prod_{i,j,k} exp \left( 2\pi i \frac{2^{i+j+k}}{N}x_i x_j y_k \right).
\end{align}
This bit-wise multiplication can be expressed using Pauli operators:
\begin{align}
\prod_{i,j,k}exp \left( 2\pi i \frac{2^{i+j+k-3}}{N}(1-\sigma_z^{(i)}) (1-\sigma_z^{(j)})(1-\sigma_z^{(k)}) \right)\label{eq:u_in_pauli}
\end{align}.

We then organize the operators into three terms:
\begin{align}
U(x,y)=\prod_{i,j,k}exp \left( \alpha_{i,j,k}\sigma_z^{(i)}\sigma_z^{(j)}\sigma_z^{(k)} \right)\prod_{i,j}exp \left( \beta_{i,j}\sigma_z^{(i)}\sigma_z^{(j)} \right) \prod_{i}exp \left( \gamma_{i}\sigma_z^{(i)} \right)\label{eq:ccz_as_zzz}.
\end{align}
We use $\alpha$'s, $\beta$'s, and $\gamma$'s to represent the phases generated by these terms, which can be calculated from Eq.\ref{eq:u_in_pauli}.
The third term contains single-qubit z-rotations that are implemented efficiently as software-phase-advances. The zz-interactions in the second term, are implemented as XX-gates sandwiched between single qubit rotations. The first term includes three body zzz-interactions, which can be decomposed using zz-interactions using the following relation:
\begin{align}
exp(-\pi/4 i \sigma_y^{(i)}\sigma_y^{(j)}) exp(i \theta \sigma_x^{(j)}\sigma_x^{(k)})exp(i\pi/4 \sigma_y^{(i)}\sigma_y^{(j)})=exp(-i \theta \sigma_y^{(i)}\sigma_z^{(j)}\sigma_x^{(k)})
\end{align}

This decomposition enables efficient construction of the following cascade of zzz-interactions: 
\begin{align}
exp(-i \theta_1 \sigma_y^{(a)}\sigma_z^{(b)}\sigma_x^{(1)})exp(-i \theta_2 \sigma_y^{(a)}\sigma_z^{(b)}\sigma_x^{(2)})...exp(-i \theta_n \sigma_y^{(a)}\sigma_z^{(b)}\sigma_x^{(n)})=\\
exp(-\pi/4 i \sigma_y^{(a)}\sigma_y^{(b)}) exp(i \theta_1 \sigma_x^{(b)}\sigma_x^{(1)}).exp(i \theta_2 \sigma_x^{(b)}\sigma_x^{(2)})....exp(i \theta_n \sigma_x^{(b)}\sigma_x^{(n)})exp(i\pi/4 \sigma_y^{(a)}\sigma_y^{(b)})
\end{align},
which are efficiently implemented from the native xx-interaction and single-qubit rotations.

Using the above decomposition, we can implement the first term in Eq.~\ref{eq:ccz_as_zzz} using the circuit shown Fig.~\ref{fig:X2ModN_circuit_detail}.

\begin{figure}[htbp]
\centering
 \includegraphics[width=0.9\textwidth]{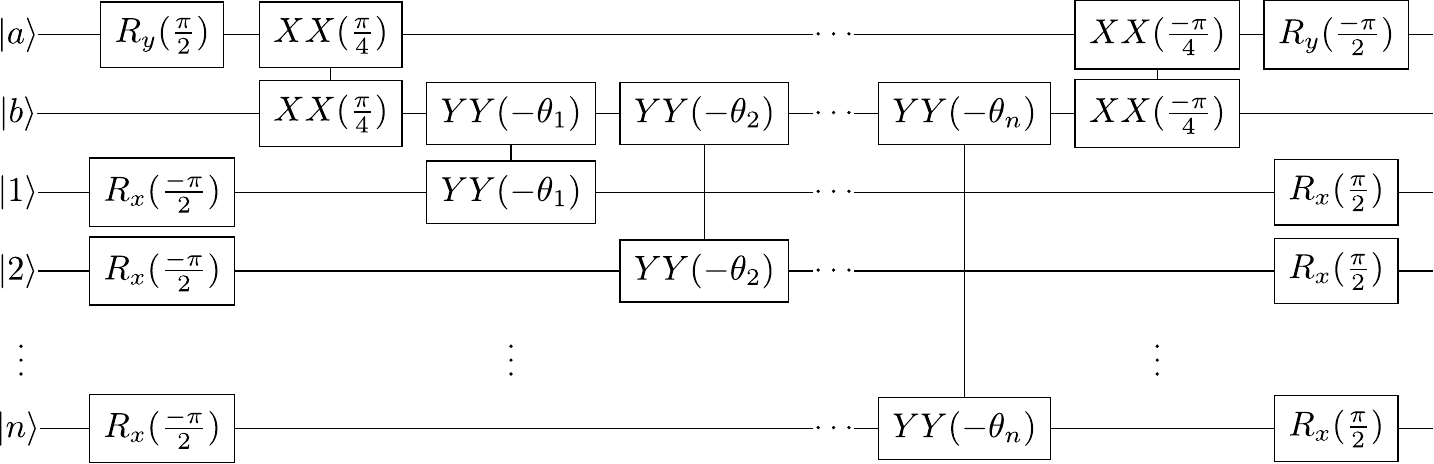}
 \caption{Circuit implementing the first term in Eq.~\ref{eq:ccz_as_zzz}.}\label{fig:X2ModN_circuit_detail}
\end{figure}

With this term implemented, we complete the construction of the full unitary U(x,y). After applying the unitary, we obtain the state
\begin{align}
     \alpha \sum_{0\leq x \leq N/2}\ket{x}\sum_{0\leq y \leq N} e^{2 \pi i \frac{x^2 y \, mod\, N}{N}} \ket{y}.
\end{align}

We then apply the inverse Quantum Fourier Transforma  $QFT^\dagger$ to the $y$-register, which gives us:
\begin{align}
     \alpha \sum_{0\leq x \leq N/2}\ket{x}\ket{y=x^2\, mod\, N}
\end{align}
Next we measure the $y$-register to find an output w, and the $x$-register contains the superposition of a colliding input pair.

The number of qubits used to represent $y$ in experiments are 3,4,4 and 5 for N=8, N=15, N=16 and N=21, respectively. The number of qubits used to represent $x$ in experiments equals the length of the $r$ string in table~\ref{tab:postselection_X2modN}.

\subsection{Circuit construction of the LWE based protocol}
\label{circ_construct_lwe}

In this section we describe the procedure for implementing the circuit $U(A, b, x, y)$, displayed in Fig.~\ref{fig:all_combined}(d).
First, let us comment on the parameters on which this unitary depends.
The matrix $A \in \mathbb{Z}^{m \times n}_q$ and vector $s \in \{0,1\}^n$ are sampled uniformly at random by the verifier\footnote{Technically, the matrix $A$ is sampled together with the TCF trapdoor. However, as explained in~\cite{brakerski2018cryptographic}, the distribution from which the matrix is sampled is statistically close to a uniform distribution over $\mathbb{Z}^{m \times n}_q$}. The vector $e \in \mathbb{Z}^m_q$ is sampled from a discrete Gaussian distribution (see Brakerski et al.~\cite{brakerski2018cryptographic} for more details on the parameter choices). The verifier constructs $y = As + e \in \mathbb{Z}^m_q$ and sends $A$ and $y$ to the prover.

Upon receiving $A$ and $y$, the prover must evaluate the function $f(b, x) = \lfloor Ax + b \cdot y \rceil$ in superposition, where $\lfloor \cdot \rceil$ denotes a rounding operation corresponding to taking the most significant bit of each component in the vector $Ax + b \cdot y$. 
It should be noted that this specific function, which uses rounding, differs from the TCF used by Brakerski et al.~\cite{brakerski2018cryptographic}, but is nevertheless still a TCF~\cite{liu2021depth}.

To perform the coherent evaluation, the prover will use three registers (for the $b$ and $x$ inputs, as well as for the output of the TCF) to create the superposition state as well as a fourth ancilla register, which will be used to perform the unitary $U(A, b, x, y)$. The prover starts by applying a layer of Hadamard gates to all input qubits and the ancilla register (that were initialized as $\ket{0}$). The resulting state will be
\begin{align}
    \sum_{b\in \{0,1\}} \sum_{x\in \mathbb{Z}_q^n} \sum_{a \in \mathbb{Z}_q} \alpha \ket{b}\ket{x}\ket{a}\ket{0}
\end{align}
for some normalization constant $\alpha$ and
where the third register is the ancilla register and the last register is the output register.
In this output register, the prover must coherently add $\lfloor Ax + b \cdot y \rceil$.
As $Ax + b \cdot y$ is an $m$-component vector, we will explain the prover's operations, at a high level, for each component of the vector. For the $i$'th component of this vector, the prover first computes the modulo $q$ inner product between the $i$'th row of $A$ and $x$ and places the result in the ancilla register. Since the prover has a classical description of $A$, this will involve a series of controlled operations between the $x$ register and the ancilla register. Similar to the factoring case, this arithmetic operation is easiest to perform in the Fourier basis, hence why the ancilla register was Hadamarded. Once the inner product has been computed, the prover will perform a controlled operation between the $b$ qubit and the ancilla register in order to add the $i$'th component of $y$.
Finally, the prover will ``copy'' the most significant bit of the result into the output register. This is done via another controlled operation. The prover then uncomputes the result in the ancilla, clearing that register. In this way, the $i$'th component of $\lfloor Ax + b \cdot y \rceil$ has been added into the output register.
Repeating this procedure for all components will yield the desired state
\begin{align}
    \sum_{b\in \{0,1\}} \sum_{x\in \mathbb{Z}_q^n} \alpha' \ket{b}\ket{x}\ket{0}\ket{\lfloor Ax + b \cdot y \rceil}
\end{align}
with normalization constant $\alpha'$.

Having given the high level description, let us now discuss in more detail the specific circuits of the current implementation. From the above analysis, we can see that the total number of qubits is $N = 1+n\log_2(q) + \log_2(q) + m$. In the instance for this experiment, we chose $m=4, n=2, q=4$, resulting in $N=11$ qubits.
The first register contains $\ket{b}$ which requires only one qubit. In the second register, the vector $x = (x_0, x_1)$ consists of two components modulo $4$, which is encoded in binary with four qubits as $\ket{x} = \ket{x_{11},x_{12},x_{21},x_{22}}$. The third register, the ancilla, is one modulo $4$ component and will thus consist of two qubits. Lastly, in the fourth register, we store the result of evaluating the function, which requires another four qubits.
As mentioned, the matrix $A$ and the vector $y$ are specified classically. In the experiment, we considered four different input configurations, corresponding to four different choices for $A$, $s$ and $e$. These choices are explicitly described later in the appendix.

\begin{figure}[htbp]
    \centering
    \includegraphics[width=0.9\textwidth]{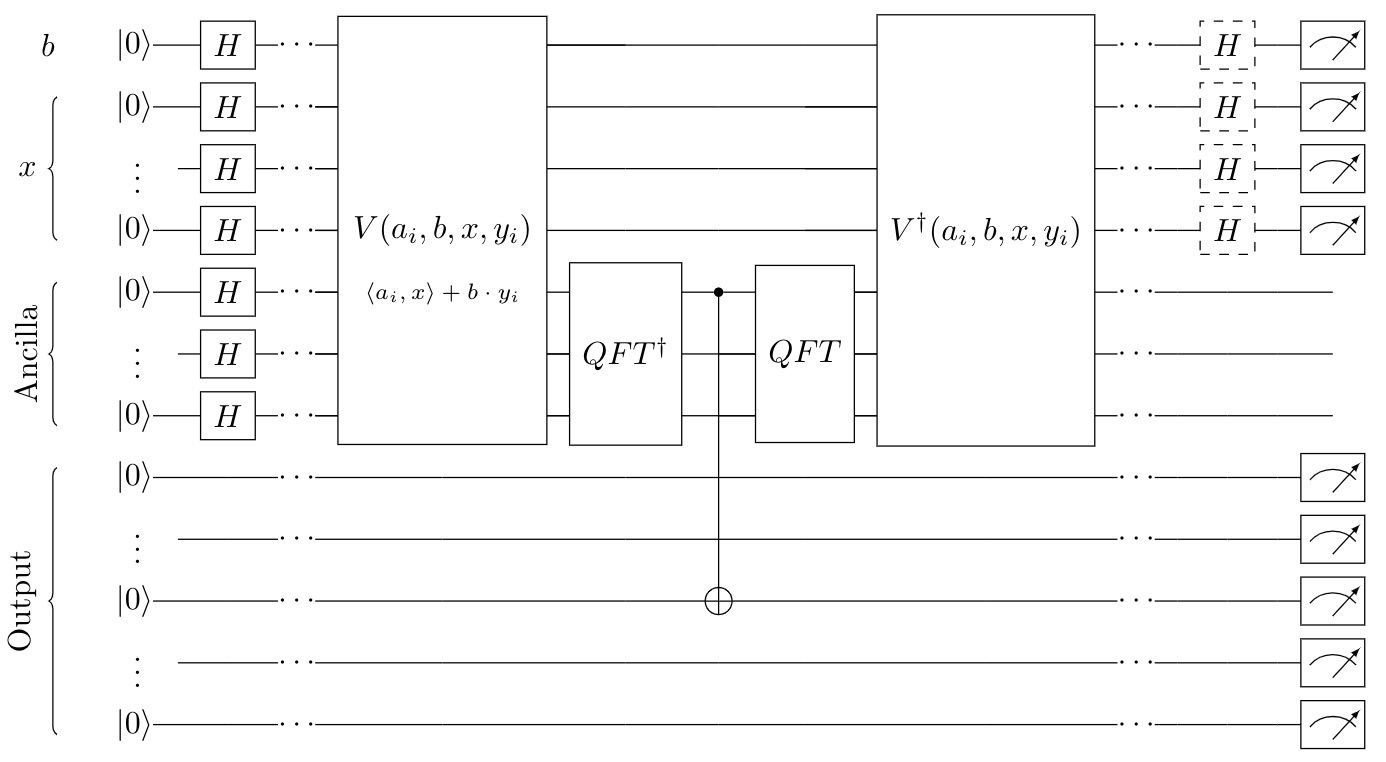}
    \caption{Circuit used to ``copy'' the most significant bit of the result from the ancilla into the output register, adding the $i$'th component of $\lfloor Ax + b \cdot y \rceil$. Here, $V$ represents the unitary used to compute $\langle a_i, x \rangle + b \cdot y_i$, in modular arithmetic, for the $i$'th row $a_i$ of the matrix $A$. Additionally, $y_i$ denotes the $i$'th entry of the vector $y = As + e$. Also, note that the target qubit of the $\mathsf{CNOT}$ in the diagram is the $i$'th qubit. The step shown is repeated for each row of $A$, indexed by $i$.}
    \label{fig:lwe_circuit_detail_msb}
\end{figure}

\begin{figure}[htbp]
    \centering
    \includegraphics[width=0.9\textwidth]{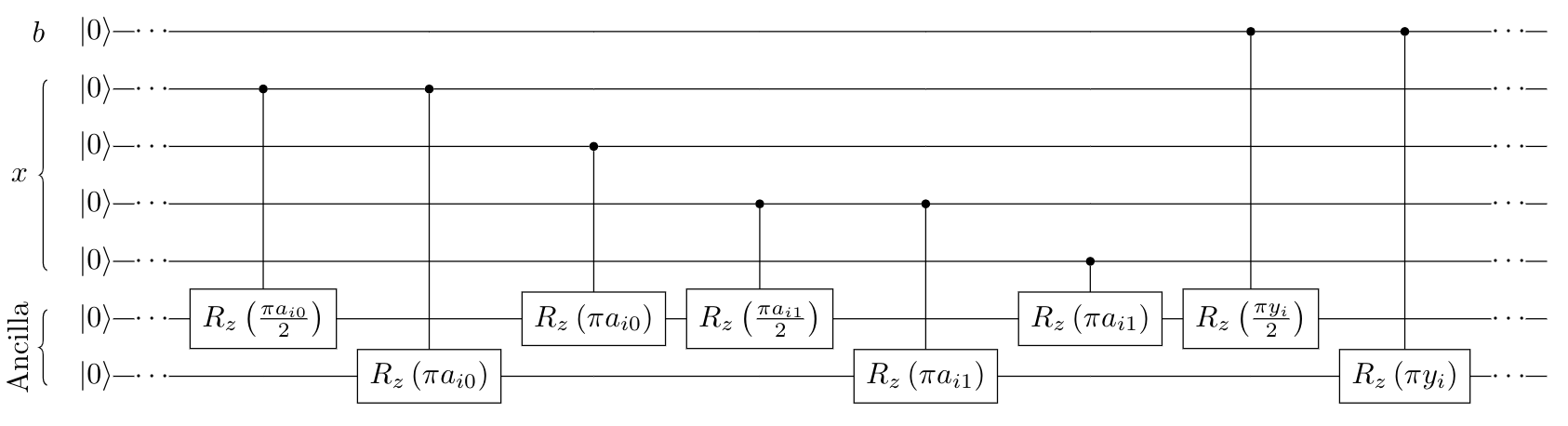}
    \caption{The explicit rotation gates used to implement the unitary $V$ from fig.\ref{fig:lwe_circuit_detail_msb} for the case of $q = 4$. Here, $a_{ij}$ denotes the entry in the $i$'th row and $j$'th column of the matrix $A$ and $y_i$ denotes the $i$'th entry of the vector $y = As + e$. The output register is omitted as there are no operations performed on it in this section of the circuit. The step shown is repeated for each row of $A$, indexed by $i$.}
    \label{fig:lwe_circuit_detail_rotations}
\end{figure}

To detail the operations implemented, as discussed previously, the prover first puts the ancilla register into the Fourier basis using the Quantum Fourier Transform (QFT). This allows them to more easily compute $\langle a_i, x\rangle + b \cdot y_i$ in the ancilla register, where $a_i$ is the $i$'th row of the matrix $A$ and $\langle \cdot, \cdot \rangle$ denotes the inner product modulo $q$. The explicit rotation gates to compute this in the Fourier basis are given in Fig.~\ref{fig:lwe_circuit_detail_rotations}. After computing this for one row $a_i$, the prover converts the ancilla back into the computational basis and ``copies'' the most significant bit stored in the ancilla register into the output register, using a $\mathsf{CNOT}$ gate, to compute the rounding function. This completes the evaluation of the function for one bit. In order to reuse the qubits in the ancilla register, the prover then reverses this computation and repeats for each row of the matrix $A$. This process of evaluating the function and reversing that computation is depicted in Fig.~\ref{fig:lwe_circuit_detail_msb}.

Finally, after completing the evaluation of the TCF, the prover measures the output register to recover the rounded result of $\lfloor Ax + b \cdot y \rceil$ for a certain value of $x$. The prover will then measure the $b$ and $x$ registers in either the $Z$ basis or $X$ basis, according to the challenge issued by the verifier. Should the verifier choose to measure in $X$ basis, the prover applies Hadamard gates on all qubits in the $b$ and $x$ registers before measuring in the computational basis.

\subsection{Instances of LWE Implemented}
\label{suppinfo:lweinstances}

\begin{table}[htbp]
    \centering
    \begin{tabular}{c|c|c|c}
       Instance & $A^\intercal$ & $e^\intercal$ & $(As + e)^\intercal$\\
       \hline\hline
       0 & $\begin{pmatrix} 0 & 2 & 0 & 1\\2 & 0 & 1 & 2\end{pmatrix}$ & $\begin{pmatrix} 0 & 1 & 0 & 0 \end{pmatrix}$ & $\begin{pmatrix} 0 & 3 & 0 & 1 \end{pmatrix}$\\
       \hline
       1 & $\begin{pmatrix} 0 & 2 & 3 & 2\\2 & 3 & 0 & 0 \end{pmatrix}$ & $\begin{pmatrix} 0 & 0 & 0 & 1\end{pmatrix}$ & $\begin{pmatrix} 0 & 2 & 3 & 3 \end{pmatrix}$\\
       \hline
       2 & $\begin{pmatrix} 2 & 0 & 0 & 1\\0 & 3 & 2 & 1\end{pmatrix}$ & $\begin{pmatrix} 1  & 0 & 1 & 0\end{pmatrix}$ & $\begin{pmatrix} 3 & 0 & 1 & 1 \end{pmatrix}$\\
       \hline
       3 & $\begin{pmatrix} 0 & 1 & 3 & 0\\3 & 0 & 0 & 2\end{pmatrix}$ & $\begin{pmatrix} 1 & 0 & 1 & 0\end{pmatrix}$ & $\begin{pmatrix} 0 & 1 & 3 & 1 \end{pmatrix}$\\
    \end{tabular}
    \caption{Details of the LWE instances. Note that the entries are transposed and for all instances we use $s^\intercal = \begin{pmatrix} 0 & 1\end{pmatrix}$.}
    \label{tab:lwe_instances}
\end{table}

Here, we explicitly detail the LWE instances that were used in the experiment. Recall that such an instance is defined by $A, s,$ and $e$, where $A \in \mathbb{Z}_q^{m\times n}$, $s \in \{0,1\}^n,$ and $e \in \mathbb{Z}_q^m$ for integers $m, n, q \in \mathbb{Z}$. In this experiment, we used $m = 4, n = 2, q = 4$ for all of the instances. Table~\ref{tab:lwe_instances} displays the explicit matrices and vectors used.

\subsection{Derivation of Quantum-Classical Threshold for LWE based Protocol} \label{suppinfo:threshold}

\begin{prop}
For any classical prover, the probabilities that they pass branches A and B, $p_A$ and $p_B$, must obey the relation
\begin{equation}
    p_A + 2 p_B - 2 < \epsilon(\lambda)
\end{equation}
where $\epsilon$ is a negligible function of the security parameter $\lambda$.
\end{prop}

\begin{proof}
We first want to find the probability that the classical prover both responds correctly for Branch A and, for the same image $w$ that they committed to the verifier, Branch B is also correct with probability greater than $1/2 + \mu(\lambda)$, where $\mu$ is a non-negligible function of the security parameter $\lambda$. Let this second probability be denoted as
\begin{equation}
    p_{\mathrm{good}} \equiv \Pr_w [p_{B,w} > 1/2 + \mu(\lambda)]
\end{equation}
By a union bound, we arrive at a bound on the desired probability
\begin{equation}
    \label{eq:initialprob}
    \Pr[A \text{ correct and } p_{B,w} > 1/2 + \mu(\lambda)] > p_A + p_{\mathrm{good}} - 1
\end{equation}
Now, we wish to write $p_{\mathrm{good}}$ in terms of $p_B$. Let $S$ be the set of $w$ values for which $p_{B, w} > 1/2 + \mu(\lambda)$. By definition, we know that with probability $p_{\mathrm{good}}$, the prover samples a $w \in S$ so that they pass the verifier's Branch B test with probability at least $1/2 + \mu(\lambda)$ and at most $1$. Similarly, we know that with probability $1 - p_{\mathrm{good}}$, the prover samples a $w \notin S$ so that they pass the verifier's Branch B test with probability at most $1/2$. Hence, overall we see that the probability that the prover passes Branch B is at most the convex mixture of these two cases, i.e.
\begin{equation}
    p_B < 1 \cdot p_{\mathrm{good}} + 0.5 \cdot (1 - p_{\mathrm{good}})
\end{equation}
Solving for $p_{\mathrm{good}}$, we then obtain
\begin{equation}
    p_{\mathrm{good}} > 2p_B - 1
\end{equation}
Substituting this into Equation \ref{eq:initialprob}, we have
\begin{equation}
    \label{eq:breakahb}
    \Pr[A \text{ correct and } p_{B,w} > 1/2 + \mu(\lambda)] > p_A + 2p_B - 2
\end{equation}
However, notice that this probability on the left hand side is the probability of breaking the adapative hardcore bit property, which we know \cite{brakerski2018cryptographic} must have
\begin{equation}
    \Pr[A \text{ correct and } p_{B,w} > 1/2 + \mu(\lambda)] < \epsilon(\lambda)
\end{equation}
where $\epsilon$ is a negligible function. Thus, combining this with Equation \ref{eq:breakahb}, we obtain the desired inequality
\begin{equation}
    p_A + 2p_B - 2 < \epsilon(\lambda)
\end{equation}
\end{proof}

\subsection{Computation of statistical significance contours}
\label{suppinfo:significance}

Here we describe the computation of the contour lines denoting various levels of statistical significance in Figure~\ref{fig:all_combined}(b,e) of the main text.
Recall the probabilities $p_A$ and $p_B$ introduced in Section~\ref{sec:threshold}, which denote a prover's probability of passing the standard basis and interference test, respectively.
Assuming the cryptographic soundness of the claw-free property of the TCF, and in the limit of large problem size, any classical cheating strategy must have true values of $p_A^c$ and $p_B^c$ that obey the bound $p_A^c + 2p_B^c - 2 < 0$ for the LWE protocol and $p_A^c + 4p_B^c - 4 < 0$ for the factoring-based protocol.
To find the statistical significance of a pair of values $p_A$ and $p_B$ measured from an (ostensibly) quantum prover, we consider the null hypothesis that the data was generated by a classical cheater (which obeys the bounds above), and compute the probability that the given data could be generated by that null hypothesis.
In particular, since the bounds above exclude a region of a two-dimensional space, we consider an infinite ``family'' of null hypotheses which lie along the boundary, and define the overall statistical significance of measuring $p_A$ and $p_B$ to be the \emph{minimum} of the statistical significances across the entire family of null hypotheses---that is, we define it as the significance with respect to the \emph{least rejected} null hypothesis.

To compute the statistical significance of a result $(p_A, p_B)$ with respect to a particular null hypothesis $(p_A^c, p_B^c)$, we define the ``quantumness'' $q$ of an experiment as $q(p_A, p_B) = p_A + 4p_B - 4$ for the factoring-based protocol and $q(p_A, p_B) = p_A + 2p_B - 2$ for the LWE protocol.
Letting $N_A$ and $N_B$ be the number of experimental runs performed for each branch respectively, we define the joint probability mass function (PMF) as the product of the PMFs of two binomial distributions $B(N_A, p_A^c)$ and $B(N_B, p_B^c)$. Mathematically the joint PMF is
\begin{equation}
f(k_A, k_B; p_A^c, p_B^c, N_A, N_B) = \binom{N_A}{k_A} \binom{N_B}{k_B} (p_A^c)^{k_A} (p_B^c)^{k_B} (1-p_A^c)^{N_A-k_A} (1-p_B^c)^{N_B-k_B}
\end{equation}
where $k_A = p_A N_A$ and $k_B = p_B N_B$ are the ``count'' of passing runs for each branch respectively.
Finally, we compute the statistical significance of a result $(p_A, p_B)$ as the probability of achieving quantumness measure of at least $q' = q(p_A, p_B)$.
Under a null hypothesis $(p_A^c, p_B^c)$, this is the sum of the PMF over all $k_A, k_B$ for which $q(k_A/N_A, k_B/N_B) > q'$.

In practice, for the contour lines of Figure~\ref{fig:all_combined}(b,e), we begin with a desired level of statistical significance (say, $5 \sigma$), and given the sample sizes $N_A$ and $N_B$ we compute the value of $q'$ that would achieve at least that significance over all null hypotheses inside the classical bound.

\subsection{Relative performance of additional instances of factoring-based protocols implemented with delayed-measurement.}

In Fig.~\ref{fig:all_combined}(f), we show the relative performance of the factoring-based protocol for $N=8$, performed both interactively and with delayed measurement. In Fig.~\ref{fig:X2ModN_Extra} we display the relative performance for $N \in \{15, 16, 21\}$ (for which experiments were run with delayed measurement only).

\begin{figure}[htbp]
\centering
 \includegraphics[width=0.7\textwidth]{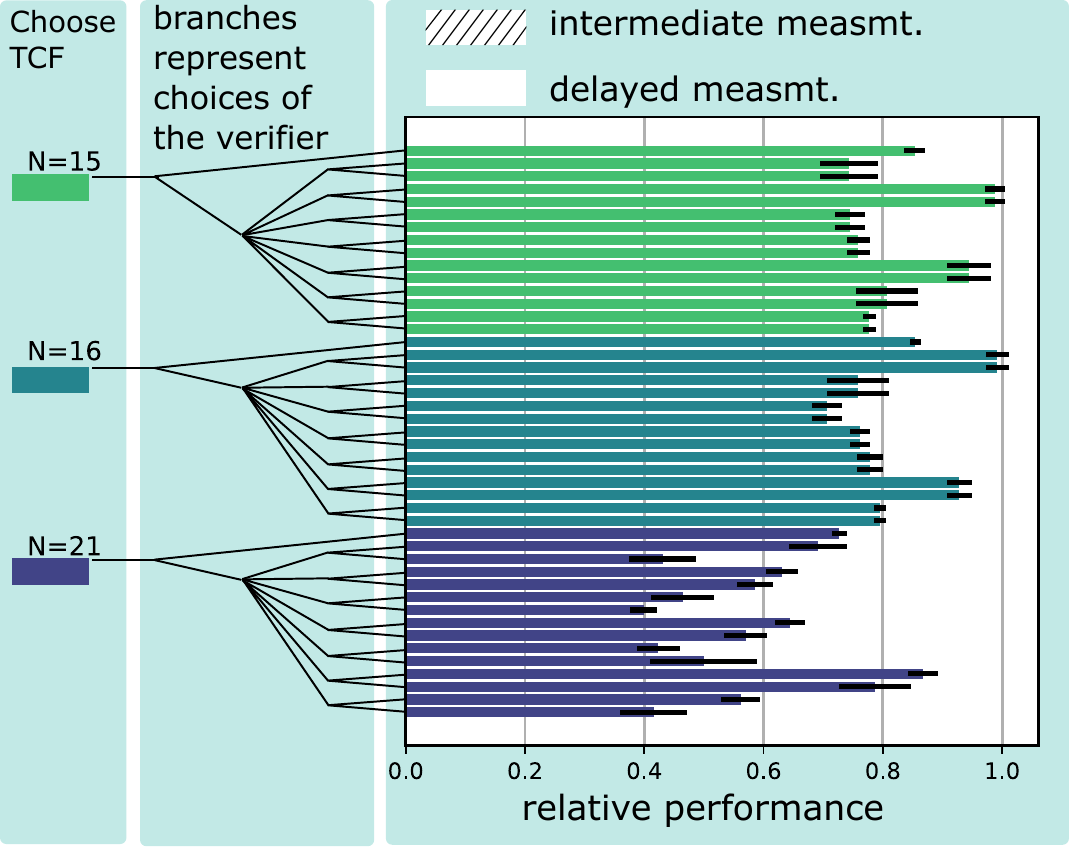}
 \caption{Extra instances of factoring-based protocols implemented with delayed measurement.}\label{fig:X2ModN_Extra}
\end{figure}

\subsection{Estimate of resources required to achieve quantum advantage}
\label{suppinfo:estimates}

For a conclusive demonstration of quantum advantage, we desire the quantum machine to perform the protocol significantly faster than the amount of time a classical supercomputer would require to break the trapdoor claw-free function---ideally, orders of magnitude faster.
To achieve this, we must set the parameters of the cryptographic problem to sufficiently large values.
A major benefit of using protocols based off of established cryptographic assumptions (like factoring and LWE) is that the classical hardness of breaking these assumptions is extremely well studied, due to the implications for security.
Thus the most straightforward way to choose parameters for our tests is to rely on publicly-available recommendations for \emph{cryptographically secure} key sizes, which are used in practice.
These parameter settings are designed to be not just slow for classical machines, but infeasible even for classical machines years from now---and thus certainly would constitute a definitive demonstration of quantum advantage.
However, setting the parameters to these values may be considered overkill for our purposes, especially since we'd like the problem size to be as small as possible in order to make the protocols maximally feasible on near term quantum devices.
With these considerations, in this section we provide two estimates for each protocol: we begin by providing estimates for smaller problem sizes that still would demonstrate some level of quantum advantage, and then give estimates based on cryptographic parameters.

We conservatively estimate that a future quantum device running the protocols investigated in this work at scale would complete the protocols on a time scale of at most hours.
Thus, to demonstrate quantum advantage by several orders of magnitude, we desire to set the parameters such that a classical supercomputer would require time on the order of thousands of hours to break the TCF.
In 2020, Boudot et al. reported the record-breaking factorization of a 795-bit semiprime~\cite{boudot_comparing_2020}.
The cost of the computation was about 1000 core-years, meaning that a 1000-core cluster would complete it in a year.
We consider this sufficient cost to demonstrate quantum advantage.
We emphasize also that factoring is one of the most well-studied hard computational problems; the record of Boudot et al. is the product of decades of algorithm development and optimization and thus it is unlikely that any innovations will drastically affect the classical hardness of factoring in the near term.
The computational Bell test protocol using a 795-bit prime could be performed using only about 800 qubits by computing and measuring the bits of the output value $w$ one-by-one; however the gate count and circuit depth can be dramatically reduced by explicitly storing the full output value $w$, requiring roughly 1600 qubits total~\cite{kahanamoku2021classically}.
Because it is so much more efficient in gate count, we use the 1600 qubit estimate as the space requirement to demonstrate quantum advantage with the computational Bell test protocol.

For LWE, estimating parameters for the same level of hardness (1000 core-years) is difficult to do exactly, because to our knowledge that amount of computational resources has never been applied to breaking an LWE instance.
However, we may make a rough estimate.
There is an online challenge (\url{https://www.latticechallenge.org/lwe_challenge/challenge.php}) intended to explore the practical classical hardness of LWE, in which users compete for who can break the largest possible instance.
As of this writing, the largest instances which have been solved use LWE vectors of about 500-1000 bits (depending on the noise level of the error vector), but the computational cost of these calculations was only of order 0.5 core-years.
To require 1000 core-years of computation time, we estimate that the LWE vectors would need to be perhaps 1000-2000 bits in length; by not explicitly storing the output vector $w$ but computing it element-by-element (similar in principle to the scheme for evaluating $x^2 \bmod N$ using only $\log(N) + 1$ qubits~\cite{kahanamoku2021classically}) it may be possible to perform the LWE protocol using a comparable number of qubits to the bit length of one LWE vector.

We now provide estimates for cryptographic parameters; that is, parameters for which it expected to be completely infeasible for a classical machine to break the trapdoor claw-free function. 
For the factoring-based protocol, we may apply NIST's recommended key sizes for the RSA cryptosystem, whose security relies on integer factorization.
NIST recommends choosing a modulus $N$ with length 2048 bits.
By using circuits optimized to conserve qubits, it is possible to evaluate the function $x^2 \bmod N$ using only $\log(N)+1$ qubits, yielding a total qubit requirement of 2049 qubits~\cite{kahanamoku2021classically}.
However, the circuit depth can be improved significantly by including more qubits; a more efficient circuit can be achieved with roughly $2\log(N) \sim 4100$ qubits. 
Because LWE is not yet broadly used in practice like RSA is, NIST does not provide recommendations for key sizes in its documentation.
However, we can use the estimates of Lindner and Peikert\cite{lindner_better_2011} to find parameters which are expected to be infeasible classically.
In Fig.~3 of that work, the authors suggest using LWE vectors in $\mathbb{Z}_q^n$ with $n=256$ and $q=4093$ for a ``medium'' level of security.
Vectors with these parameters are $n \log (q) \sim 3072$ bits long.
To store both an input and output vector would thus require roughly $\sim 6200$ qubits.
By repeatedly reusing a set of qubits to compute the output vector element-by-element the computation could be performed using roughly 3100 qubits.

\newpage


%

\end{document}